\documentclass[runningheads,a4paper]{llncs}

\usepackage{amssymb}
\usepackage{graphicx}
\usepackage{comment}
\usepackage{enumitem}
\usepackage{lipsum}
\usepackage{ntheorem}
\usepackage{caption}
\usepackage{lipsum}

\usepackage{hyperref}
\newcommand{\orcid}[1]{\href{https://orcid.org/#1}{\includegraphics[width=8pt]{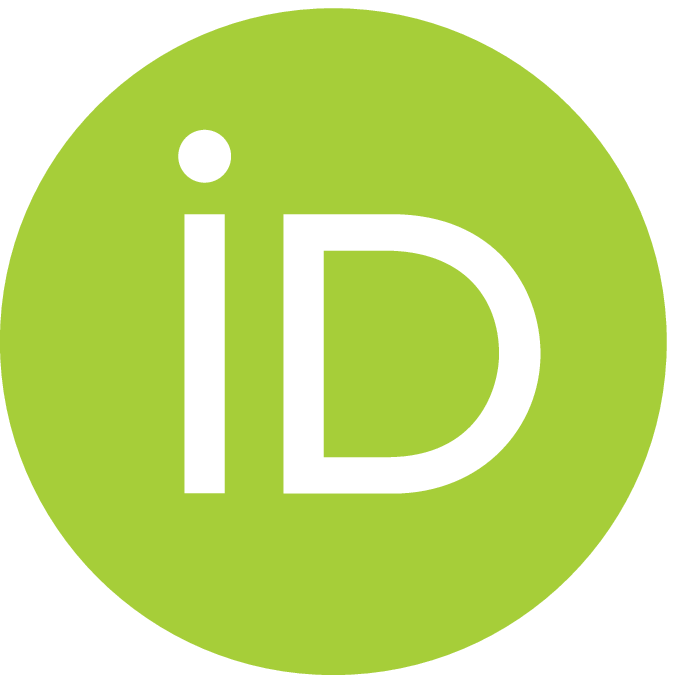}}}

\usepackage{url}
\urldef{\mailsa}\path|shenmaier@mail.ru|

\theorembodyfont{\upshape}
\newtheorem{rem}{Remark}
\theorembodyfont{\itshape}
\newtheorem*{definition*}{Definition.}

\AtBeginDocument{%
  \setlength{\oddsidemargin}{\dimexpr(\paperwidth-\textwidth)/2-1in}%
  \setlength{\evensidemargin}{\oddsidemargin}%
  \setlength{\topmargin}{%
    \dimexpr(\paperheight-\textheight)/2-\headheight-\headsep-1in}%
}

\renewenvironment{abstract}{
  \small
  \list{}{
    \setlength{\leftmargin}{.5cm}%
    \setlength{\rightmargin}{\leftmargin}%
  }%
  \item\relax
}

\begin{document}

\mainmatter  

\title{Asymptotic Optimality of the Greedy Patching Heuristic for Max~TSP in Doubling Metrics}

\titlerunning{Greedy Patching Heuristic for Max~TSP in doubling metrics}

%
%
\author{Vladimir Shenmaier \orcid{0000-0002-4692-1994}}
%
\authorrunning{V.\,V. Shenmaier}

\institute{Sobolev Institute of Mathematics, Novosibirsk, Russia\\
\mailsa}

%
%

\tocauthor{Vladimir Shenmaier}
\maketitle

\begin{abstract}{\bf Abstract.}
The maximum traveling salesman problem (Max~TSP) consists of finding a Hamiltonian cycle with the maximum total weight of the edges in a given complete weighted graph.
We prove that, in the case when the edge weights are induced by a metric space of bounded doubling dimension, asymptotically optimal solutions of the problem can be found by the simple greedy patching heuristic.
Taking as a start point a maximum-weight cycle cover, this heuristic iteratively patches pairs of its cycles into one minimizing the weight loss at each step.
\begin{keywords}
Max~TSP $\cdot$ Heuristic $\cdot$ Greedy patching $\cdot$ Doubling metric $\cdot$ Asymptotically exact algorithm
\end{keywords}
\end{abstract}

\section{Introduction}
The maximum traveling salesman problem can be formulated as follows:\medskip

\noindent\textbf{Max~TSP.}
Given an $n$-vertex complete weighted (directed or undirected) graph $G$ with non-negative edge weights, find a Hamiltonian cycle in $G$ with the maximum total weight of the edges.\medskip

Max~TSP is the maximization version of the classic traveling salesman problem (TSP) and, like TSP, is among the most intensively researched NP-hard problems in computer science.
In this paper, we consider the metric Max~TSP, i.e., the special case in which the edge weights satisfy the triangle inequality and the symmetry axiom.\medskip

\noindent\textbf{Related work.}
Max~TSP has been actively studied since the 1970s.
The approximation factors of currently best polynomial-time algorithms in different cases are: $2/3$ for arbitrary asymmetric weights \cite{Kaplan}; $7/9$ for arbitrary symmetric weights \cite{Paluch79}; $35/44$ for the asymmetric metric case \cite{Kowalik3544}; and $7/8$ for the metric case~\cite{Kowalik78}.

On the complexity side, Max~TSP is APX-hard even in a metric space with distances $1$ and $2$: It follows from the corresponding result for TSP \cite{Papadimitriou1993,EK}.
The problem remains NP-hard in the geometric setting when the vertices of the input graph are some points in space $\mathbb R^3$ and the distances between them are induced by Euclidean norm~\cite{Fekete}.
The proof of this fact implies that the Euclidean Max~TSP does not admit a scheme FPTAS unless P$=$NP.
However, Max~TSP admits a scheme EPTAS in an arbitrary metric space of fixed doubling dimension $dim$~\cite{Shen2021}.
The time complexity of this scheme is $O\big(2^{(2/\varepsilon)^{2dim+1}}+n^3\big)$.
An actual question is developing practically usable approximation algorithms.

In~\cite{Serdyukov1987}, an $O(n^3)$-time algorithm is proposed which computes asymptotically optimal solutions of Max~TSP in Euclidean space of any fixed dimension.
The relative error of this algorithm is estimated as $c_d/n^{\frac{2}{d+1}}$, where $d$ is the dimension of space and $c_d$ is some constant depending on $d$.
In~\cite{Shen2010,Shen2014}, this result is extended to the case when the edge weights are induced by any (unknown) vector norm.
The algorithms from \cite{Serdyukov1997,Barvinok} allow to find close-to-optimal and optimal solutions of Max~TSP in the case of metrics defined by polyhedrons with a small number of facets.
In~\cite{Shen2021}, an $O(n^3)$-time approximation algorithm was proposed which computes asymptotically optimal solutions of the problem in an arbitrary metric space of bounded doubling dimension $dim$.
The relative error of this algorithm is estimated as $(11/6)/n^{\frac{1}{2dim+1}}$.\medskip

\noindent\textbf{Greedy patching heuristic.}
In this paper, we address one of the simplest and natural ideas how to find good solutions of Max~TSP.
It can be described as follows: Taking as a start point a maximum-weight cycle cover, iteratively patch pairs of its cycles into one minimizing the weight loss at each step.
To patch two cycles into one, a pair of edges from different cycles is replaced by another pair of edges which connect these cycles.
For the classic TSP, the similar greedy patching algorithm was studied in~\cite{Zverovich,Goldengorin} as the ``greedy Karp-Steele heuristic'' but no theoretical estimates for its efficiency were obtained.\medskip

\noindent\textbf{Our contributions.}
We study approximation properties of the greedy patching heuristic applied to Max~TSP.
It is easy to show (Corollary~\ref{general}) that, in the general metric setting, this heuristic has a constant-factor approximation ratio.
We prove that, for instances of Max~TSP in any metric space of doubling dimension $dim$, the greedy patching heuristic finds approximate solutions of the problem with relative error at most $(7/3+o(1))/n^{\frac{1}{2dim+1}}$ as $n\rightarrow\infty$.
Thereby this simple heuristic implements an asymptotically exact algorithm in the case of fixed or sublogarithmic doubling dimensions, i.e., when $dim=o(\log n)$.

An advantage of the greedy patching heuristic over the algorithms from~\cite{Shen2021} is that it does not require any information about the value of doubling dimension $dim$, which may not always be available, even approximately.
It should be noted that the derived theoretical estimate for the relative error of this heuristic is rather rough and may be improved, especially in geometric cases.

\section{Basic definitions and properties}
A \emph{metric space} is an arbitrary set ${\cal M}$ with a non-negative distance function $dist$ which is defined for each pair $x,y\in{\cal M}$ and satisfies the triangle inequality and the symmetry axiom.
Given a metric space $({\cal M},dist)$, a \emph{ball} of radius $r$ in this space centered at a point $x\in{\cal M}$ is the set $B(x,r)=\{y\in{\cal M}\,|\,dist(x,y)\le r\}$.
The \emph{doubling dimension} of a metric space is the smallest value $dim\ge 0$ such that every ball in this space can be covered by $2^{dim}$ balls of half the radius.

\begin{rem}\label{rem1}
It is easy to see that, if a metric space is of doubling dimension at most $dim$, then each $r$-radius ball in this space can be covered by $(2/\delta)^{dim}$ balls of radius $\delta r$, where $\delta$ is any value from $(0,1)$ (e.g., see~\cite{Shen2021}).
\end{rem}

Suppose that we are given a set $V$ of $n$ points in ${\cal M}$ and also all the pairwise distances $dist(a,b)$, $a,b\in V$.
Denote by $G[V]$ the complete weighted undirected graph on the vertex set $V$ in which the weight of every edge $\{a,b\}$ is defined as $dist(a,b)$.
The metric Max~TSP asks to find a maximum-weight Hamiltonian cycle in $G[V]$.

\begin{definition*}
Let $c_1$, $c_2$ be vertex-disjoint cycles in $G[V]$ and $\{a_i,b_i\}$ be any edge in $c_i$, $i=1,2$.
A \emph{patch of the cycles $c_1,c_2$ on the edges $\{a_1,b_1\}$, $\{a_2,b_2\}$} is a combining of these cycles into one by replacing the pair of edges $\{a_1,b_1\}$, $\{a_2,b_2\}$ by one of two pairs $\{a_1,b_2\}$, $\{a_2,b_1\}$ or\, $\{a_1,a_2\}$, $\{b_1,b_2\}$ of the maximum total weight.
A \emph{weight loss} of this patch is the value
\begin{eqnarray*}
Loss\big(\{a_1,b_1\},\{a_2,b_2\}\big)=dist(a_1,b_1)+dist(a_2,b_2)-\\
\max\big\{dist(a_1,b_2)+dist(a_2,b_1),\,dist(a_1,a_2)+dist(b_1,b_2)\big\}.
\end{eqnarray*}
\end{definition*}

\begin{definition*}
A \emph{cycle cover} of a graph is a spanning subgraph of this graph in which every connected component is a simple cycle.
\end{definition*}

\noindent\textbf{Greedy Patching Heuristic (GPH).}

\noindent\emph{Input}: a set $V$ of $n$ points in ${\cal M}$; the distances $dist(a,b)$ for all $a,b\in V$.
\emph{Output}: a Hamiltonian cycle $H$ in the graph $G[V]$.\medskip

\noindent\emph{Initial Step}:
By using the $O(n^3)$-time algorithm from \cite{Gabow}, find a maximum-weight cycle cover $C_0$ of the graph $G[V]$.\medskip

\noindent\emph{Patching Steps}:
Let $C=C_0$ and, while $C$ contains more than one cycles, repeat the fol\-lo\-wing operations.
Find edges $\{a_1,b_1\}$ and $\{a_2,b_2\}$ from different cycles of $C$ with the minimum value of $Loss\big(\{a_1,b_1\},\{a_2,b_2\}\big)$.
Patch the cycle cover $C$ by replacing the pair of edges $\{a_1,b_1\}$, $\{a_2,b_2\}$ by one of the pairs $\{a_1,a_2\}$, $\{b_1,b_2\}$ or $\{a_1,b_2\}$, $\{a_2,b_1\}$ with the maximum total weight.

\begin{lemma}\label{loss1}
The weight loss at each patch of GPH is at most $w(C)/n\le w(C_0)/n$, where $w(C)$ and $w(C_0)$ are the total weights of $C$ and $C_0$, respectively.
\end{lemma}

\begin{proof}
Let $\tau$ be a lightest edge in the current cycle cover $C$.
Then the weight $w(\tau)$ of this edge is at most $w(C)/n$. 
On the other hand, the triangle inequality easily implies that, for any edge $u$ in $C$, the value of $Loss(u,\tau)$ is at most $w(\tau)$.
So GPH can always choose a patch with weight loss at most $w(C)/n$.
The lemma is proved.
\hfill$\Box$
\end{proof}

\begin{corollary}\label{general}
In the general metric setting, the approximation ratio of GPH is at least $e^{-1/3}$.
\end{corollary}

\begin{proof}
By Lemma~\ref{loss1}, the approximation ratio of GPH is at least
$(1-1/n)^{k-1}$, where $k$ is the number of cycles in $C_0$.
But, obviously, $k\le n/3$, so we obtain the estimate $(1-1/n)^{n/3-1}\ge e^{-1/3}$.
\hfill$\Box$
\end{proof}

\begin{lemma}\label{loss2}
For any edges $\{a,a'\}$, $\{b,b'\}$ from two vertex-disjoint cycles in $G[V]$, we have $Loss\big(\{a,a'\},\{b,b'\}\big)\le 2dist(a,b)$.
\end{lemma}

\begin{proof}
Indeed, by the definition of weight loss and by the axioms of metric, we have
\begin{eqnarray*}
Loss\big(\{a,a'\},\{b,b'\}\big)\le dist(a,a')+dist(b,b')-dist(a,b')-dist(b,a')\le 2dist(a,b).
\end{eqnarray*}
The lemma is proved.
\hfill$\Box$
\end{proof}

\section{Justification of the Greedy Patching Heuristic}
\begin{theorem}\label{th1}
If the space $({\cal M},dist)$ is of doubling dimension at most $dim$, then the relative error of GPH is at most $(7/3+o(1))/n^{\frac{1}{2dim+1}}$ as $n\rightarrow\infty$.
\end{theorem}

\begin{proof}
Let $\{a_0, b_0\}$ be a shortest edge in $C_0$ and $t_0=dist(a_0,b_0)$.
Then, obviously, we have $\displaystyle t_0\le\frac{w(C_0)}{n}$.
Further, we will use a real-value parameter $\rho\in(0,1)$ to be specify later.
Define the value $\displaystyle R_0=\frac{w(C_0)}{n\rho}$ and consider the ball $B(a_0,R_0)$, which will play an important role in justifying GPH.



\begin{lemma}\label{R0}
Let $c$ be a cycle in $C_0$ none of whose vertices lie in the ball $B(a_0,R_0)$ and let $\{a,b\}$ be any edge in $c$.
Then $dist(a,b)\ge 2R_0-2t_0$.
\end{lemma}

\begin{proof}
Suppose that $dist(a,b)<2R_0-2t_0$.
Then, replacing the edges $\{a_0,b_0\}$, $\{a,b\}$ by $\{a_0,a\}$, $\{b_0,b\}$, we obtain a new cycle cover of $G[V]$ whose total weight is at least
\begin{eqnarray*}
w(C_0)+dist(a_0,a)+dist(b_0,b)-t_0-dist(a,b)\ge\\
w(C_0)+dist(a_0,a)+dist(a_0,b)-2t_0-dist(a,b)\ge\\
w(C_0)+2R_0-2t_0-dist(a,b)>w(C_0),
\end{eqnarray*}
which contradicts the choice of the cycle cover $C_0$.
The lemma is proved.
\hfill$\Box$
\end{proof}

\begin{definition*}
A cycle in $C$ which doesn't intersect the ball $B(a_0,R_0)$, i.e., none of whose vertices belongs to $B(a_0,R_0)$, will be referred to as \emph{far}.
A cycle which contains at least one vertex in $B(a_0,R_0)$ will be referred to as \emph{near}.
\end{definition*}

We will use the following denotation: we assume that, at the $i$th patching step, $i=1,2,\dots$, GPH replaces edges $e_{2i-1}$, $e_{2i}$ from cycles $c_{2i-1}$, $c_{2i}$ in $C$.

%

Let us divide all the patching steps in GPH into groups (types) I, II, and III:

The group I consists of the patches for which at least one of cycles $c_{2i-1}$, $c_{2i}$ is far.

The group II consists of the patches for which both cycles $c_{2i-1}$, $c_{2i}$ are near and $Loss(e_{2i-1},e_{2i})\le\displaystyle\frac{2\delta w(C_0)}{n}$, where $\delta\in(0,1)$ is another real-value parameter, additionally to $\rho$, which will be specified later.

The group III consists of all the other patches, i.e., those for which both cycles $c_{2i-1}$, $c_{2i}$ are near and $Loss(e_{2i-1},e_{2i})>\displaystyle\frac{2\delta w(C_0)}{n}$.

Now, let us estimate the total weight loss at each of these groups.
Denote by $K_{\rm I}$, $K_{\rm II}$, and $K_{\rm III}$ the numbers of patches in GPH of the types I, II, and III, respectively.

\begin{lemma}\label{KI}
$\displaystyle K_{\rm I}\le\frac{n\rho}{6(1-\rho)}$.
\end{lemma}

\begin{proof}
Obviously, each I-type patch reduces the number of far cycles in $C$ by $1$.
All the other patches do not increase this number.
So the the number of far cycles in the original cycle cover $C_0$ is at least $K_{\rm I}$.

By Lemma~\ref{R0}, every edge of a far cycle is of weight at least
$$2R_0-2t_0\ge\frac{2w(C_0)}{n}\Big(\frac{1}{\rho}-1\Big),$$
therefore, the total weight of all the far cycles in $C_0$ is at least $\displaystyle K_{\rm I}\frac{6w(C_0)}{n}\Big(\frac{1}{\rho}-1\Big)$.
But, on the other hand, this weight is at most $w(C_0)$, so we have $\displaystyle K_{\rm I}\le\frac{n\rho}{6(1-\rho)}$.
The lemma is proved.
\hfill$\Box$
\end{proof}

\begin{lemma}\label{KIII}
$\displaystyle K_{\rm III}\le\Big(\frac{4}{\rho\delta}\Big)^{dim}$.
\end{lemma}

\begin{proof}
Obviously, each III-type patch reduces the number of near cycles in $C$ by $1$.
All the other patches do not increase this number.
So, by the time of the first III-type patch, the current cycle cover $C$ contains at least $K_{\rm III}$ cycles which intersect the ball $B(a_0,R_0)$.
But $\displaystyle R_0=\frac{w(C_0)}{n\rho}$, so the ball $B(a_0,R_0)$ can be covered by $\displaystyle\Big(\frac{4}{\rho\delta}\Big)^{dim}$ balls of the radius $\displaystyle r=\frac{w(C_0)\delta}{2n}$ (see Remark~\ref{rem1}).
Thus, if $\displaystyle K_{\rm III}>\Big(\frac{4}{\rho\delta}\Big)^{dim}$, then there exist at least two vertices, say $a$ and $b$, from different cycles in $C$ which hit the same $r$-radius ball.

Denote by $a'$ and $b'$ vertices in $C$ adjacent to $a$ and $b$, respectively.
Then, by Lemma~\ref{loss2} and by the triangle inequality, we have
$$Loss\big(\{a,a'\},\{b,b'\}\big)\le 2dist(a,b)\le 4r=\frac{2w(C_0)\delta}{n}$$
It follows that, by the time of the first III-type patch, there exists a patch of the type II, which contradicts the rules of GPH.
The lemma is proved.
\hfill$\Box$
\end{proof}

For the value of $K_{\rm II}$, we will use the obvious estimate $K_{\rm II}\le n/3$.

By Lemmas \ref{KI} and~\ref{loss1}, the total weight loss at I-type patches is at most $\displaystyle\frac{w(C_0)\rho}{6(1-\rho)}$.
Next, since $K_{\rm II}\le n/3$, the total weight loss at II-type patches is at most $\displaystyle\frac{2\delta w(C_0)}{3}$.
Finally, by Lemmas \ref{KIII} and~\ref{loss1}, the total weight loss at III-type patches is at most $\displaystyle\Big(\frac{4}{\rho\delta}\Big)^{dim}\frac{w(C_0)}{n}$.
Thus, the relative error of GPH is
$$err\le\frac{\rho}{6(1-\rho)}+\frac{2\delta}{3}+\Big(\frac{4}{\rho\delta}\Big)^{dim}/n.$$

It remains to select good values of the parameters $\rho$ and $\delta$.
Let $\delta=\displaystyle\frac{1}{n^{1/(2dim+1)}}$ and $\rho=4\delta$.
If $n^{1/(2dim+1)}\ge 8$, then $\rho\le 1/2$ and $\displaystyle\frac{1}{1-\rho}=1+o(1)$ as $\rho\rightarrow 0$, so 
$$err\le\frac{4+o(1)}{6n^{1/(2dim+1)}}+\frac{2}{3n^{1/(2dim+1)}}+\frac{1}{n^{1/(2dim+1)}}=\frac{7/3+o(1)}{n^{1/(2dim+1)}}$$
as $n\rightarrow\infty$.
If $n^{1/(2dim+1)}<8$, then we use a simple estimate based on Corollary~\ref{general}: 
$$err\le 1-e^{-1/3}\approx 0.2835<7/24<\frac{7/3}{n^{1/(2dim+1)}}.$$
The theorem is proved.
\hfill$\Box$
\end{proof}

\section{Conclusion}
We prove that the simple greedy patching heuristic gives asymptotically optimal solutions of Max~TSP in doubling metrics.
An interesting direction for future work is comparing this heuristic with the asymptotically exact algorithm from~\cite{Shen2021} in practice, e.g., on random data.
Another possible subject of further investigation is improving the estimate for the relative error of GPH in the case of $d$-dimensional Euclidean space $(\mathbb R^d,\ell_2)$.
Our hypothesis is that this estimate is $O\big(1/n^{\frac{2}{d+1}}\big)$.

\vspace{2em}\noindent\textbf{Acknowledgments.}
The study was carried out within the framework of the state contract of the Sobolev Institute of Mathematics (project FWNF-2022-0019).



\begin{thebibliography}{}



\bibitem{Barvinok}
Barvinok A., Fekete S.P., Johnson D.S., Tamir A., Woeginger~G.J., Woodroofe~R.:
The geometric maximum traveling salesman problem.
J.~ACM \textbf{50}(5), 641--664 (2003)
    



\bibitem{EK}
Engebretsen L., Karpinski M.:
TSP with bounded metrics.
J.~Comp. System Sci. \textbf{72}(4), 509--546 (2006)
    
\bibitem{Fekete}
Fekete S.P.:
Simplicity and hardness of the maximum traveling salesman problem under geometric distances.
In: Proc. 10th ACM-SIAM Symposium on Discrete Algorithms (SODA~1999), 337--345 (2015)


\bibitem{Gabow}
Gabow H.:
An efficient reduction technique for degree-constrained subgraph and bidirected network flow problems.
In: Proc. 15th ACM Symposium on Theory of Computing (STOC~1983), 448--456 (1983)

\bibitem{Zverovich}
Glover F., Gutin G., Yeo A., Zverovich A.:
Construction heuristics for the asymmetric TSP.
European J. Operational Research \textbf{129}(3), 555--568 (2001)

\bibitem{Goldengorin}
Goldengorin B., J\"ager G., Molitor P.:
Tolerance Based Contract-or-Patch Heuristic for the Asymmetric TSP.
In: Proc. 3rd Workshop on Combinatorial and Algorithmic Aspects of Networking (CAAN~2006), Lecture Notes in Computer Science, vol.~4235, 86--97 (2006)


\bibitem{Kaplan}
Kaplan H., Lewenstein M., Shafrir N., Sviridenko M.:
Approximation algorithms for asymmetric TSP by decomposing directed regular multigraphs.
J.~ACM \textbf{52}(4), 602--626 (2005)




\bibitem{Kowalik78}
Kowalik L., Mucha M.:
Deterministic $7/8$-approximation for the metric maximum TSP.
Theor. Comp. Sci. \textbf{410}(47--49), 5000--5009 (2009)

\bibitem{Kowalik3544}
Kowalik L., Mucha M.:
$35/44$-approximation for asymmetric maximum TSP with triangle ine\-qua\-li\-ty.
Algorithmica \textbf{59}(2), 240--255 (2011)


\bibitem{Paluch79}
Paluch K., Mucha M., M\c{a}dry A.:
A $7/9$-approximation algorithm for the maximum tra\-ve\-ling salesman problem.
In: Proc. 12th Workshop on Approximation Algorithms for Combinatorial Optimization (APPROX~2009), Lecture Notes in Computer Science, vol.~5687, 298--311 (2009)


\bibitem{Papadimitriou1993}
Papadimitriou C.H., Yannakakis M.:
The traveling salesman problem with distances one and two.
Math. Oper. Res. \textbf{18}(1), 1--11 (1993)




\bibitem{Serdyukov1987}
Serdyukov A.I.:
An asymptotically exact algorithm for the traveling salesman problem for a maximum in Euclidean space (in Russian).
Upravlyaemye sistemy \textbf{27}, 79--87 (1987)



\bibitem{Serdyukov1997}
Serdyukov A.I.:
The maximum-weight traveling salesman problem in finite-dimensional real spaces.
In: Operations Research and Discrete Analysis, Mathematics and Its Applications, vol.~391, pp.~233--239. Kluwer Academic Publishers, Dordrecht (1997)

\bibitem{Shen2010}
Shenmaier V.V.:
An asymptotically exact algorithm for the maximum traveling salesman problem in a finite-dimensional normed space.
J.~Appl. Industr. Math. \textbf{5}(2), 296--300 (2011)

\bibitem{Shen2014}
Shenmaier V.V.:
Asymptotically optimal algorithms for geometric Max TSP and Max $m$-PSP.
Discrete Appl. Math. \textbf{163}(2), 214--219 (2014)

\bibitem{Shen2021}
Shenmaier V.V.:
Efficient PTAS for the maximum traveling salesman problem in a metric space of fixed doubling dimension.
Optimization Letters (2021)

\end{thebibliography}
\end{document}